\documentclass[copyright,creativecommons]{eptcs}
\usepackage{breakurl}             % Not needed if you use pdflatex only.

\usepackage{amsmath}
\usepackage{amsthm}
\usepackage{amssymb}

\usepackage{enumerate}
\usepackage{hyperref}
\usepackage{xspace}
\usepackage{stmaryrd}

\usepackage{rotating}

\usepackage{gastex}
\usepackage{amssymb}
\usepackage{booktabs}

\newtheorem{example}{Example}

\newtheorem{lemma}{Lemma}

\title{Expressiveness of Visibly Pushdown Transducers\thanks{This work has been supported by the PEPS project SOSP (``Synthesis of Stream Processors'') funded by CNRS and by the project ECSPER (ANR-09-JCJC-0069) funded by the ANR.}}

\author{Mathieu Caralp$\quad$Pierre-Alain Reynier$\quad$ Jean-Marc Talbot
\institute{Laboratoire d'Informatique Fondamentale de Marseille \\
Aix-Marseille Universit\'e \& CNRS, France}
\email{\{mathieu.caralp,pierre-alain.reynier,jean-marc.talbot\}@lif.univ-mrs.fr}
\and
Emmanuel Filiot\thanks{FNRS Research Associate (``Chercheur Qualifi\'e'')}
\institute{CS Department \\ Universit\'e Libre de Bruxelles, Belgium}
\email{efiliot@ulb.ac.be}
\and
Fr\'ed\'eric Servais
\institute{Hasselt University and transnational University of Limburg\\ Belgium}
\email{frederic.servais@gmail.com}
}

\def\titlerunning{Expressiveness of Visibly Pushdown Transducers}
\def\authorrunning{M. Caralp, E. Filiot, P.-A. Reynier, F. Servais and J.-M. Talbot}

\newcommand{\ptime}{{\sc PTime}\xspace}

\newcommand{\exptimecomplete}{{\sc ExpTime\text{-}c}\xspace}
\newcommand{\nptime}{{\sc NPTime}\xspace}

\newcommand{\mytimes}{\mathord\times}

\newcommand{\pts}{\ensuremath{\mathsf{PTs}}\xspace}

\newcommand{\vpa}{\ensuremath{\mathsf{VPA}}\xspace}

\newcommand{\vpt}{\ensuremath{\mathsf{VPT}}\xspace}
\newcommand{\vpts}{\ensuremath{\mathsf{VPTs}}\xspace}

\newcommand{\wnvpt}{\ensuremath{\mathsf{wnVPT}}\xspace}
\newcommand{\wnvpts}{\ensuremath{\mathsf{wnVPTs}}\xspace}

\newcommand{\htos}{\ensuremath{\mathsf{H2S}}\xspace}
\newcommand{\htostr}{\ensuremath{\mathsf{H2S_{tr}}}\xspace}
\newcommand{\htoh}{\ensuremath{\mathsf{H2H}}\xspace}
\newcommand{\htohtr}{\ensuremath{\mathsf{H2H_{tr}}}\xspace}
\newcommand{\htob}{\ensuremath{\mathsf{H2B}}\xspace}

\newcommand{\hedgeset}{\ensuremath{\mathcal{H}}}
\newcommand{\wnset}{\ensuremath{\mathcal{W}}}
\newcommand{\bwnset}{\ensuremath{\mathcal{BW}}}
\newcommand{\btset}{\ensuremath{\mathcal{B}}}
\newcommand{\fcns}{\ensuremath{\mathsf{fcns}}\xspace}

\newcommand{\lin}{\ensuremath{\mathsf{lin}}\xspace}
\newcommand{\hedge}{\ensuremath{\mathsf{hedge}}\xspace}

\newcommand{\interp}[1]{\llbracket #1 \rrbracket}

\begin{document}
\maketitle

\begin{abstract}
Visibly pushdown transducers (VPTs) are visibly pushdown
automata extended with outputs. They have been introduced
to model transformations of nested words, i.e. words with a
call/return structure. As trees and more generally hedges can be 
linearized into (well) nested words, VPTs are a natural formalism to express
tree transformations evaluated in streaming. This paper aims at characterizing 
precisely the expressive power of VPTs with respect to other tree transducer models.
\end{abstract}

\section{Introduction}
\label{sec:introduction}

Visibly pushdown machines~\cite{RM04}, automata (\vpa) or transducers, 
are pushdown machines such that stack behavior is synchronized with
the structure of the input word. Precisely, the input alphabet is
partitioned into call and return symbols.
When reading a call symbol the machine must push a symbol onto the
stack, and when reading a return symbol it must pop a symbol from the stack. %

Visibly pushdown transducers (\vpts)~\cite{RS08,Servais11,MFCS2010,SLLN09} extend 
visibly pushdown automata \cite{RM04} with outputs. Each transition is
equipped with an output word that is appended to the output tape
whenever the transition is triggered. A \vpt thus transforms an input
word into an output word obtained as the concatenation of all the
output words produced along a
successful run on that input. \vpts are a strict subclass of pushdown
transducers (\pts) and strictly extend finite state transducers.
Several problems that are undecidable for \pts are decidable for
\vpts, most notably: functionality (in \ptime), $k$-valuedness (in \nptime) and
functional equivalence (\exptimecomplete)~\cite{MFCS2010}. \vpts are
closed by regular look-ahead which makes them a robust class of
transformations \cite{conf/sofsem/FiliotS12}.

Unranked trees and more generally hedges can be linearized into
well-nested words over a structured alphabet (such as XML documents).
\vpt are therefore a suitable formalism to express hedge
transformations. In particular, they can express operations such as node deletion, 
renaming and insertion. As they process the linearization from left to right, 
they are also an adequate formalism to model and analyze
transformations in streaming, as shown in \cite{filiot_et_al:LIPIcs:2011:3352}. \vpts output strings, therefore on
well-nested inputs they define hedge-to-string transformations, and if the output strings are
well-nested too, they define hedge-to-hedge transformations. 

In this paper, we characterize the expressive power of \vpts
w.r.t. their ability to express hedge-to-string (\htos), and
hedge-to-hedge (\htoh)
transformations. To do so, we define a top-down model of
hedge-to-string transducers, inspired by classical top-down tree
transducers. They correspond to parameter-free linear order-preserving 
macro forest transducers that output strings \cite{PerSei04}. We define a syntactic 
restriction of \htos that captures exactly \vpts, and show that if the
\vpts runs on binary encodings of hedges, then they have exactly the
same expressive power as \htos. We show that those results still hold
when both models are restricted to hedge-to-hedge transformations.
Based on those results, we compare \vpts with classical ranked tree
transducers, such as top-down tree transducers \cite{tata} and macro tree
transducers \cite{Engelfriet03}.

\section{Transducer Models for Nested Words and Hedges}
\label{sec:definitions}

%!TEX root = main.tex

\noindent \textbf{Words and Nested Words}
The set of finite words over a (finite) alphabet $\Sigma$ is denoted by $\Sigma^*$, and the empty
word is denoted by $\epsilon$. A \emph{structured alphabet} is a pair $\Sigma = (\Sigma_c,\Sigma_r)$ of
disjoint alphabets, of call and return symbols respectively. Given a structured alphabet 
$\Sigma$, we always denote by $\Sigma_c$ and $\Sigma_r$ its implicit structure, and
identify $\Sigma$ with $\Sigma_c\cup \Sigma_r$.

A \emph{nested word} is a finite word over a structured alphabet. 
The set of \emph{well-nested words} over a structured alphabet $\Sigma$ is the least set,
denoted by $\wnset_{\Sigma}$, that satisfies $(i)$ $\epsilon\in \wnset_\Sigma$, $(ii)$ 
for all $w,w'\in \wnset_\Sigma$, $ww'\in \wnset_\Sigma$ (closure under concatenation), and
$(iii)$ for all $w\in\wnset_\Sigma$, $c\in \Sigma_c$, $r\in \Sigma_r$, $cwr\in\wnset_\Sigma$. E.g. on $\Sigma = (\{c_1, c_2\}, \{ r\})$, the nested word $c_1rc_2r$ is well-nested while
$rc_1$ is not. Finally, note that any well-nested word $w$ is either empty or can be 
decomposed uniquely as $w = cw_1rw_2$ where $c\in\Sigma_c, r\in\Sigma_r$, $w_1,w_2\in \wnset_\Sigma$. 

\vspace{2mm}
\noindent \textbf{Hedges}
Let $\Lambda$ be an alphabet. We let $S(\Lambda)$ be the signature
$\{0, \cdot\}\cup \{a\ |\ a\in \Lambda\}$ where $0$ is a constant
symbol, $a\in\Lambda$ are unary symbols and $\cdot$ is a binary
symbol. The set of \emph{hedges} $\hedgeset_\Lambda$ over $\Lambda$ is
the quotient of the free $S(\Lambda)$-algebra by the associativity of
$\cdot$ and the axioms $0\cdot h = h\cdot 0 = h$. The constant $0$ is
called the empty hedge. We may write $a$ instead of $a(0)$, and omit
$\cdot$ when it is clear from the context. \emph{Unranked trees} are
particular hedges of the form $a(h)$ where $h\in
\hedgeset_\Lambda$. 
%The set of unranked trees overs $\Lambda$ is denoted by $\unrankedtreeset_\Lambda$.
Note that any hedge $h$ is either empty or can be decomposed as $h = a(h_1)\cdot h_2$.

Hedges over $\Lambda$ can be naturally encoded as well-nested words over the
structured alphabet $\Lambda_s = (\Lambda_c, \Lambda_r)$ where $\Lambda_c$ and
$\Lambda_r$ are new alphabets respectively defined by
$\Lambda_c = \{ c_a\ |\ a\in \Lambda\}$ and $\Lambda_r = \{ r_a\ |\ a\in \Lambda\}$. 
This correspondence is given via a morphism $\lin : \hedgeset_\Lambda\rightarrow \wnset_{\Lambda_s}$ 
inductively defined by: $\lin(0) = \epsilon$ and $\lin(a(h_1).h_2) = c_a\lin(h_1)r_a\lin(h_2)$. E.g. for $\Lambda = \{ a, b\}$, we have $\lin(ab(ab)) = c_ar_ac_bc_ar_ac_br_br_b$.

Conversely, any well-nested word over a structured alphabet $\Sigma$ can be encoded as
an hedge over the product alphabet $\Sigma_c\times \Sigma_r$, via 
the mapping $\hedge: \wnset_{\Sigma}\rightarrow \hedgeset_{\Sigma_c\times \Sigma_r}$ defined
as $\hedge(\epsilon) = 0$ and 
$\hedge(c w_1 r w_2) = (c,r)(\hedge(w_1))\cdot \hedge(w_2)$ for all $(c,r)\in \Sigma_c\times \Sigma_r$ and
all $w_1,w_2\in \wnset_{\Sigma}$.

% The \emph{height} of an hedge is defined inductively as: 
% $\height(0)= 0$, 
% $\height(t\cdot h)=\max(\height(t),\height(h))$ and 
% $\height( f(h))=1+\height(h)$ where $f\in\Sigma$.

% \subsection{Transductions}

\vspace{2mm}
\noindent\textbf{Binary Trees}
We consider here an alphabet $\Lambda$ augmented with some special
symbol $\bot$. We define the set of \emph{binary trees}
$\btset_\Lambda$ as a particular case of unranked trees over $\Lambda
\cup \{\bot\}$. Binary trees are defined recursively as:
$(i)$ $\bot \in \btset_\Lambda$, and $(ii)$ for all $f \in \Lambda$, if
$t_1,t_2 \in \btset_\Lambda$ then $f(t_1\,t_2) \in \btset_\Lambda$.  

There is a well-known correspondence between hedges and binary trees by
means of an encoding called the first-child next-sibling
encoding. This encoding is given by the mapping \fcns defined as: 
$(i)$ $\fcns(0) = \bot$, $(ii)$ $\fcns(f(h_1)h_2) = f(\fcns(h_1) \,
\fcns(h_2))$ for all $h_1,h_2$ in $\hedgeset_\Lambda$. 

The strong relationship between hedges and well-nested words can be
considered when restricted to binary trees: we define
$\bwnset_{\Lambda_s}$ the set of binary well-nested words over
the structured alphabet $(\Lambda_c\cup \{\bot_c\},\Lambda_r\cup
\{\bot_r\})$ as the least set satisfying: $(i)$ $\bot_c\bot_r \in
\bwnset_{\Lambda_s}$ and $(ii)$ for all $f_c\in \Lambda_c, \ f_r \in
\Lambda_r$, if $w^{b}_1 , w^{b}_2 \in \bwnset_{\Lambda_s}$ then $f_c
\, w^{b}_1 \, w^{b}_2 \, f_r \in \bwnset_{\Lambda_s}$.  Note that the morphism
$\lin$ applied on binary trees from $\btset_\Lambda$ yields binary
nested words in $\bwnset_{\Lambda_s}$.

% Obviously, the morphism $\lin$ still applies on binary trees from
% $\btset_\Lambda$ yielding binary nested words from
% $\bwnset_{\Lambda_s}$. Conversely, applying the morphism $\hedge$ over
% binary nested words from $(\Sigma_c \cup \{\bot_c\}, \Sigma_r \cup 
% \{\bot_r\})$ yields binary trees from $(\Sigma_c \cup \{\bot_c\}, \Sigma_r \cup 
% \{\bot_r\})$.

% JE NE SAIS PAS SI ON DOIT METTRE DES ARITES .... 

% Hedges and thus, binary trees may be represented as well-nested words
% by means of the $\lin$ application. 
Finally, we can define the first-child next-sibling encoding of hedges
as binary trees, directly on linearizations; consider a structured
alphabet $\Sigma$ extended as $\Sigma_\bot =
(\Sigma_c\cup\{\bot_c\},\Sigma_r\cup\{\bot_r\})$. For all well-nested
words $w$ over $\Sigma$, we define $\fcns(w)$ over the alphabet
$\Sigma_\bot$ recursively as $(i)$ $\fcns(\epsilon)= \bot_c \bot_r$
and $(ii)$ $\fcns(cw_1r w_2) = c \, \fcns(w_1) \, \fcns(w_2) \,r$ for
all $w_1,w_2 \in \wnset_\Sigma$.

\vspace{2mm}
\noindent\textbf{Visibly Pushdown Transducers}
Let $\Sigma$ be a structured alphabet, and $\Delta$ be an alphabet. 
A \emph{visibly pushdown transducer} from $\Sigma$ to $\Delta$ (the class is denoted $\vpt(\Sigma,\Delta)$) is a tuple 
$A=(Q,I,F,\Gamma,\delta)$ where $Q$ is a finite set of states, $I\subseteq Q$ the set of initial states,
$F\subseteq Q$ the set of final states, $\Gamma$ the (finite) stack alphabet, $\bot\notin \Gamma$
  is the bottom stack symbol, and $\delta=\delta_c \uplus \delta_r$ is the
  transition relation where:
\begin{itemize}
\item $\delta_c \subseteq Q \times\Sigma_c \times
\Gamma\times \Delta^* \times Q$ are the {\em call transitions}, 
\item $\delta_r \subseteq Q \times\Sigma_r \times\Gamma \times
  \Delta^* \times Q$ are the {\em return transitions}.
\end{itemize}

A configuration of $A$ is a pair $(q,\sigma)$ where $q\in Q$ and
$\sigma\in\bot\cdot\Gamma^*$ is a stack content.  Let $w=a_1 \ldots
a_l$ be a (nested) word on $\Sigma$, and $(q,\sigma),(q',\sigma')$ be
two configurations of $A$.  A {\em run} of the \vpt $A$ over $w$ from
$(q,\sigma)$ to $(q',\sigma')$ is a (possibly empty) sequence of
transitions $\rho=t_1t_2\dots t_l\in\delta^*$ such that there exist
$q_0, q_1, \dots q_l\in Q$ and $\sigma_0, \dots \sigma_l\in
\bot\cdot\Gamma^*$ with $(q_0,\sigma_0)=(q,\sigma)$,
$(q_l,\sigma_l)=(q',\sigma')$, and for each $0<k\leq l$, we have
either $(i)$ $t_k=(q_{k-1},a_{k},\gamma,w_k,q_{k})\in \delta_c$ and
$\sigma_{k}=\sigma_{k-1}\gamma$, or $(ii)$ $t_k=(q_{k-1},
a_{k},\gamma,w_k, q_{k})\in \delta_r$, and
$\sigma_{k-1}=\sigma_{k}\gamma$. When the sequence of transitions is empty, 
$(q,\sigma)=(q',\sigma')$.

The \emph{output} of $\rho$ is the word $w\in \Delta^*$ defined as the
concatenation $w=w_1\ldots w_l$ when the sequence of transitions is
not empty and $\epsilon$ otherwise.  %We write
% $(q,\sigma)\xrightarrow{u\mid w} (q',\sigma')$ when there exists a run
% on $u$ from $(q,\sigma)$ to $(q',\sigma')$ producing $w$ as output.
Initial (resp. final) configurations are pairs $(q,\bot)$ with $q\in
I$ (resp. with $q\in F$).  A run is {\em accepting} if it starts in an
initial configuration and ends in a final configuration. 
% The transducer $T$ defines relation from nested words to words, denoted by
% $\transduction{T}$, and defined as the set of pairs $(u,w)\in\Sigma^*\times\Delta^*$ such that
% there exists an accepting run on $u$ producing $w$ as output. Note that since we accept by empty stack 
% and there is no return transition on empty stack, $T$ accepts only well-nested words, i.e.
% $\transduction{T}\subseteq \wnset_\Sigma\times \Delta^*$.
The transducer $A$ defines a relation from nested words to words
defined as the set of pairs $(u,w)\in\Sigma^*\times\Delta^*$ such that
there exists an accepting run on $u$ producing $w$ as output.  From
now on, we confuse the transducer and the transduction it
represents. Note that since we accept by empty stack and there is no
return transition on empty stack, $A$ accepts only well-nested words,
and thus is included into $\wnset_\Sigma\times \Delta^*$.

\vspace{2mm}
\noindent\textbf{Hedge-to-string Transducers}
We present now a model of hedge-to-string transducers (\htos) that run directly on
hedges, and is closer to classical transducers than \vpts are. In particular, this model
is a syntactic subclass of macro forest transducers (MFT) \cite{PerSei04} with no
parameters, no swapping and no copy.

% They can be understood as parameter-free \textit{macro forest
  % transducers (MFT)} \cite{PerSei04} which produce
% strings.

\iffalse
With \htos, an hedge, $f(h)h'$, is rewritten in a top down manner.
The root node of the left most tree, $f(h)$, is transformed into 
an hedge over the output alphabet according to an initial rule.
Some of the leaves of this output hedge are insertion points
for the result of the recursive application of the rules
on either the hedge $h$ of the children or on the hedge $h'$ containing the siblings. 
\fi

Let $\Lambda$ and $\Delta$ be two finite alphabets.
An \emph{hedge-to-string transducer} from $\Lambda$ to
$\Delta$ (the class is denoted $\htos(\Lambda,\Delta)$) is a tuple $T=(Q, I, \delta)$ where $Q$ is a set of states,
$I\subseteq Q$ is a set of initial states and $\delta$ is a set of
rules of the form:~\footnote{We consider linear and order-preserving rules only.}
$$
q(0)\rightarrow \epsilon\qquad\qquad 
q(f(x_1)\cdot x_2)\rightarrow w_1 q_1(x_1) w_2 q_2(x_2) w_3
$$
where $q,q_1,q_2\in Q$, $f\in\Lambda$ and $w,w_1,w_2,w_3\in \Delta^*$.
  
The semantics of $T$ is defined via 
mappings $\interp{q}:\hedgeset_\Lambda\rightarrow 2^{\Delta^*}$ for
all $q\in Q$ as follows:
$$
\begin{array}{ll}
  \interp{q}(0) & \! \! =\left\{ \begin{array}{cl}\{\epsilon\} &\text{ if }q(0)\rightarrow \epsilon \in \delta\\\emptyset& \text{ otherwise}\end{array}\right . \\%\bigcup_{q(0)\rightarrow w} \{ w \} \\
\\
  \interp{q}(f(h)\cdot h') & \! \! = \displaystyle{\bigcup_{\substack{q(f(x_1)\cdot x_2)\rightarrow \\ w_1 q_1(x_1) w_2 q_2(x_2) w_3}}}
  \!\!\!\!\!\!\!\!\!\!\!\!\!\!\!\!w_1 \cdot \interp{q_1}(h) \cdot w_2 \cdot \interp{q_2}(h') \cdot
  w_3\\
% & \qquad \text{where } \ell=w_1 q_1(x_1) w_2 q_2(x_2) w_3
\end{array}
$$

The transduction of an \htos $T=(Q,I,\delta$) is defined as the relation
$\{ (h,s) \mid \exists q\in I,\ s\in\interp{q}(h) \}$.  When $s \in\interp{q}(h)$ for some 
\htos $T$, we may say that
the computation of the \htos $T$ on the hedge $h$ leads to $q$
producing $s$.  

\noindent We say that $T$ is \emph{tail-recursive} whenever in any rule, we have
$w_3=\epsilon$. We denote by \htostr the class of tail-recursive \htos.

\begin{example}\label{ex:mirror}
Let $\Lambda$ be a finite alphabet. Consider $T_1\in \htos(\Lambda,\Lambda)$ defined by $Q=I=\{q,q'\}$  and the following rules, for all $f\in\Lambda$:
$$
q(0)\rightarrow \epsilon \qquad 
q'(0)\rightarrow \epsilon \qquad 
q(f(x_1)\cdot x_2) \rightarrow q'(x_1)q(x_2)f
$$
The domain of $T_1$ is the set of strings over $\Lambda$ (viewed as a particular case
of hedges) and  
$T_1$ defines the mirror image of strings.
\end{example}

\begin{example}
  Let $\Lambda$ be a finite alphabet and $\Lambda_s$ be its structured
  version. We define $T_2\in \htos(\Lambda,\Lambda_s)$ which can non-deterministically
  root any subhedge of the input hedge under a new symbol $\#$ and output the linearization of
the new hedge. For instance, the input tree $f(abcd)$ can be non-exhaustively translated into
the string $\lin(f(a\#(bc)d))$ or the string $\lin(f(\#(ab)\#(cd)))$. Formally, $T_2$ is defined by
  $Q=\{q_{0},q_{1},q_{2}\}$, $I=\{q_{0}\}$ and $\delta$ defined as the following set of rules: (observe that $T_2\in \htostr$)
$$
\begin{array}{ll}
q_i(0)\rightarrow \epsilon \quad \forall i\in \{0,2\}  &
q_{0}(f(x_1)\cdot x_2) \rightarrow c_f q_{0}(x_1) r_f q_{0}(x_2) \\
q_0(f(x_1) \cdot  x_2) \rightarrow c_{\#} c_f q_2(x_1) r_f r_{\#} q_0(x_2)&
q_{0}(f(x_1)\cdot x_2) \rightarrow c_{\#} c_f q_{2}(x_1) r_f q_{1}(x_2) \\
q_{1}(f(x_1)\cdot x_2) \rightarrow c_f q_{2}(x_1) r_f r_{\#} q_{0}(x_2)   &
 q_{1}(f(x_1)\cdot x_2) \rightarrow c_f q_{2}(x_1) r_f q_{1}(x_2) \\
q_{2}(f(x_1)\cdot x_2) \rightarrow c_f q_{2}(x_1) r_f q_{2}(x_2)
\end{array}
$$
\end{example}

\vspace{2mm}
\noindent\textbf{Hedge-to-hedge Transducers}
We consider now transducers running on hedges but producing
(representations of) hedges as well-nested words. We define them as
restrictions of the two models we have considered so far. 

We assume the output alphabet $\Delta$ to be structured as
$(\Delta_c,\Delta_r)$. We define an $\htos(\Lambda,\Delta)$ to be an hedge-to-hedge
transducer ($\htoh(\Lambda,\Delta)$) if any rhs $w_1 q_1(x_1) w_2
q_2(x_2) w_3$ of its transition rules satisfies $w_1w_2w_3 \in
\wnset_\Delta$. We denote $\htohtr$ the class of $\htoh$ that are
additionally tail-recursive. 

Using the direct relationship between well-nested words and hedges, we
may define hedge-to-hedge transducers by means of a restriction in the
definition of \vpt: this restriction asks the nesting level of the
input and the output words to be synchronized, that is the nesting
level of the output just before reading a call (on the input) must be
equal to the nesting level of the output just after reading the
matching return (on the input).  This simple syntactic restriction
yields a subclass of \vpts \cite{MFCS2010}.

This synchronization is enforced syntactically on stack symbols, these
symbols being shared by matching call and return transitions. 

Let $A=(Q,I,F,\Gamma,\delta) \in \vpt(\Sigma,\Delta)$. Then $A$ is
\emph{well-nested} if for all $\gamma\in \Gamma$, all transitions $(q,c,\gamma,w,q')\in \delta_c$ and
$(p,r,\gamma,w',p')\in\delta_r$,  it holds that $ww' \in \wnset_\Delta$.
We denote by $\wnvpt$ the class of well-nested \vpts.

\vspace{2mm}
\noindent\textbf{Hedge-to-binary tree Transducers}
We consider transducers running on hedges and producing
(representations of) binary trees as binary well-nested words. We
define them as restrictions of hedge-to-hedge transducers.

Let  $\Delta^\bot = (\Delta^\bot_c,\Delta^\bot_r) $ be a structured output
alphabet such that $\Delta^\bot_c,\Delta^\bot_r$ contain two special symbols 
$\bot_c,\bot_r$ respectively. We define a  $\htoh(\Lambda,\Delta^\bot)$ to be an
hedge-to-binary tree transducer ($\htob(\Lambda,\Delta^\bot)$) if any right
hand-side $w_1 q_1(x_1) w_2 q_2(x_2) w_3$ of its transition rules
satisfies $w_1 = c\, w'_1$, $w_2=w_1'' w_2'$,  $w_3=w_3'\, r$ for
some $c$ in $\Delta^\bot_c$, $r$ in $\Delta^\bot_r$,  $w'_1\bot_c\bot_r w_1''$ and 
$w_2'\bot_c\bot_rw_3'$ in $\bwnset_{\Delta^\bot}$. 

Hedge-to-binary tree transducers are close to linear and 
order-preserving top-down ranked tree transducers. They
will serve us to compare the expressiveness of $\htoh$ to this latter class
of transducers defined on the first-child next-sibling encoding of input
and output hedges.

\section{Some Results on Expressiveness}
\label{sec:results}

%!TEX root = main.tex

In the sequel, we assume that input hedges accepted by transducers are non-empty.
This restriction is done without loss of generality. 
We depict on Figure \ref{fig:compexp} the results we obtained.

\begin{figure*}
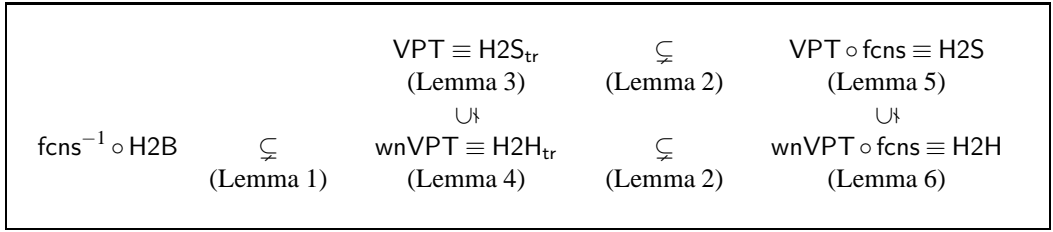


\centering
%\begin{center}
\small

\begin{tabular}{|c|}
\hline
\\

  \begin{tabular}{ccccc}
& &   \begin{tabular}{c}
        $ \vpt \equiv \htostr$
     \end{tabular}
  & $\subsetneq$ & 
     \begin{tabular}{c}
       $\vpt \circ \fcns \equiv \htos $
     \end{tabular}
   \\
& & \begin{tabular}{c}
(Lemma~\ref{lem:htostrvsvpt})\end{tabular}
& (Lemma~\ref{lem:htosvsvpt})& 
 \begin{tabular}{c}
(Lemma~\ref{lem:fcns-htostrvsvpt})
\end{tabular}
\\ 
& & 
\rotatebox[origin=c]{90}{$\subsetneq$} & & \rotatebox[origin=c]{90}{$\subsetneq$}\\
$\fcns^{-1} \circ \htob$ & $\subsetneq$ &  
\begin{tabular}{c}
$ \wnvpt 
\equiv \htohtr $
 \end{tabular}
& $\subsetneq$ & 
 \begin{tabular}{c}
 $\wnvpt \circ \fcns
\equiv \htoh$
 \end{tabular}
\\
& (Lemma~\ref{lem:htobtrvshth}) &
 \begin{tabular}{c}
(Lemma~\ref{lem:htohtrvswnvpt})\end{tabular}
& (Lemma~\ref{lem:htosvsvpt})& 
 \begin{tabular}{c}
(Lemma~\ref{lem:fcns-htohtrvswnvpt})
\end{tabular} \\ 
\end{tabular} \\ 
\\
\hline
 
\end{tabular}

\caption{\label{fig:compexp}Expressivness results in a nutshell}
%\end{center}

\end{figure*}

\subsection{Definitions of expressiveness}

Let $\Sigma$ be a structured alphabet and $\Delta$ be a finite
alphabet. We denote by ${\cal T}(\wnset_\Sigma,\Delta^*)$ the set of
transductions from $\wnset_\Sigma$ to $\Delta^*$. First observe that
the semantics of a transducer $A\in \vpt(\Sigma,\Delta)$ is an element
of ${\cal T}(\wnset_\Sigma,\Delta^*)$. Second, given a transducer
$T\in \htos(\Sigma_c\times\Sigma_r,\Delta)$, we have that $T\circ
\hedge \in {\cal T}(\wnset_\Sigma,\Delta^*)$. 
Hence, up to the mapping
$\hedge$, we can thus compare the expressiveness of a subclass ${\cal
  C}_1$ of $\vpt(\Sigma,\Delta)$ and of a subclass ${\cal C}_2$ of
$\htos(\Sigma_c\mytimes\Sigma_r,\Delta)$, by their interpretation as
transductions from $\wnset_\Sigma$ to $\Delta^*$.

Formally, given $A\in \vpt(\Sigma,\Delta)$ and $T\in
\htos(\Sigma_c\mytimes\Sigma_r,\Delta)$, we say that $A$ and $T$ are
equivalent, denoted $A \equiv T$, whenever $A = T \circ \hedge$.
Given a subclass ${\cal C}_1$ of $\vpt(\Sigma,\Delta)$ and a
subclass ${\cal C}_2$ of $\htos(\Sigma_c\mytimes\Sigma_r,\Delta)$, we
say that ${\cal C}_1$ is more expressive than ${\cal C}_2$ (resp. less
expressive), denoted ${\cal C}_1 \supseteq {\cal C}_2$ (resp. ${\cal
  C}_1 \subseteq {\cal C}_2$), whenever we have:
\begin{itemize}
\item for every $T\in {\cal C}_2$, there exists $A\in {\cal C}_1$ such that $A\equiv T$
\item for every $A\in {\cal C}_1$, there exists $T\in {\cal C}_2$ such that $A\equiv T$, respectively
\end{itemize}
Last, we write ${\cal C}_1 \equiv {\cal C}_2$ whenever ${\cal C}_1$
and ${\cal C}_2$ are expressively equivalent meaning that both
${\cal C}_1 \supseteq {\cal C}_2$ and ${\cal C}_1 \subseteq {\cal
  C}_2$ hold.

\subsection{Comparing expressiveness}

We first recall in the framework we proposed here a known
expressiveness result \cite{Servais11} comparing 
$\htoh$ and $\htob$.

\begin{lemma}\label{lem:htobtrvshth} Let $\Delta=(\Delta_c,\Delta_r)$
  and $\Delta^\bot= (\Delta_c \cup \{\bot_c\},\Delta_r \cup
  \{\bot_r\})$ be two structured alphabets. 
\begin{enumerate}
\item \label{item:lesshtobtrvshth}

For any $T \in \htob(\Lambda,\Delta^\bot)$, there exists $T' \in 
  \htoh(\Lambda,\Delta)$ such that $T'=\fcns^{-1} \circ
  T$.
\item \label{item:strictlyhtobtrvshth} 
There exists $T' \in \htoh(\Lambda,\Delta)$ such that there is no $T
\in \htob(\Lambda,\Delta^\bot)$ satisfying $T'=\fcns^{-1} \circ T$. 
\end{enumerate}
\end{lemma}

\begin{proof}
For Point (\ref{item:lesshtobtrvshth}), it is enough to apply
$\fcns^{-1}$ to the right-hand side of transition rules of $T$
(keeping sub-expressions $(q(x_i)$ unchanged) to obtain $T'$. 
For Point (\ref{item:strictlyhtobtrvshth}), for any well-nested word
$u$ let us define its size $|u|$ as the number of symbols occurring in it
and its height $||u||$ as: 
$(i)$ $||u||=0$ if $u=\epsilon$ and $(ii)$ $||cvrw||=max(1+||v||,||w||)$
if $u=cvrw$. Size and height can be defined on hedges $h$
accordingly by considering size and heigth of $\lin(h)$. The following facts
can easily be proved: (Fact 1) One can devise a transducer $T'$ that flattens
its input into a sequence ($T'(f(h_1)h_2)=c_fr_f\,T'(h_1)T'(h_2)$). Then, 
$|T'(h)|=2|h|$ and $||T'(h)||=1$. (Fact 2) For all $T$ in $\htob(\Lambda,\Delta^\bot)$,
there exists $k_T$ in $\mathbb{N}$ such that for all hedges $h$,
$||T(h)|| \le k_T ||h||$; (Fact 3) If $w \in \wnset_\Delta$, $||w||=1$ and $|w|=n$ then
$||\fcns(w)||=n$.

Now, consider the family $H_{n\mid n \in \mathbb{N}}$ of hedges $h$
such that $||h||=n$ and $|h|=2^n$. For any $h \in H_n$,  
$|T'(h)|=2^{n+1}$ and $||T'(h)||=1$. Hence, $||\fcns(T'(h))||=2^{n+1}$. 
Assuming that $T$ exists yields $||\fcns(T'(h))||=2^{n+1}=||T(h)||
  \le k_T n$ for some constant $k_T$ for all $n$. Contradiction.
\end{proof}

It turns out that \htos are stricly more expressive than \vpts. Formally:

\begin{lemma}\label{lem:htosvsvpt}
    There exists $T \in \htos(\Sigma_c\mytimes \Sigma_r,\Delta)$ such that for all $A\in\vpt(\Sigma,\Delta)$, 
    $T \circ \hedge \not\equiv A$. 
\end{lemma}

\begin{proof}
  Consider the variant over the input alphabet $\Sigma_c \times
  \Sigma_r$ of the transducer $T$ defined in Example
  \ref{ex:mirror}. It is easy to see this transducer produces an
  output (after an \hedge application) only on nested words from
  $(\Lambda_c. \Lambda_r)^*$. Over such input words, any \vpt admits
  only finitely many configurations in its accepting runs and thus, is
  equivalent to some finite state transducer.  But it is well known that
  finite state transducer can not compute the mirror image of its
  inputs.
\end{proof}

Informally, this is due to the abitility that \htos have to
"complete'' the output once the current hedge is processed. This
ability vanishes when tail-recursive \htos are considered.

\begin{lemma}\label{lem:htostrvsvpt}
    $\vpt(\Sigma,\Delta) \equiv \htostr(\Sigma_c\mathord\times \Sigma_r,\Delta)$. 
\end{lemma}

\begin{proof}[(Sketch)]
  Intuitively, in order to transform $A \in \vpt(\Sigma,\Delta)$ into
  $T \in \htostr(\Sigma_c\mathord\times \Sigma_r,\Delta) $, we proceed
  as follows. States of $T$ are pairs of states of $A$, corresponding
  to states reached respectively at the beginning and at the end of
  the processing of an hedge. More formally, the following rule will
  exist in $T$ iff there exist a call transition on $c$ from $p$ to
  $p_1$, a matching return on $r$ from $p_2$ to $q_1$, the hedge
  represented by $x_1$ (resp. by $x_2$) can be processed from state
  $p_1$ to state $p_2$ (resp. from $q_1$ to $q$):
	 $$
	 (p,q)((c,r)(x_1)\cdot x_2) \rightarrow w_1 \cdot (p_1,p_2)(x_1) \cdot w_2 \cdot (q_1,q)(x_2)
	 $$ 
	 
	 The word $w_1$ (resp. $w_2$)
	is the output of the call transition (resp. of the return transition).
 It is worth observing that this encoding 
	directly implies the tail-recursive property of $T$.
	
	The converse construction follows the same ideas. The stack is used to store the transition used on the call symbol, to recover it when reading the return symbol. 
\end{proof}

Lemma \ref{lem:htosvsvpt} still holds even if we restrict 
\htos to \htoh, because the transducer defining the transduction of Example \ref{ex:mirror} is
actually an \htoh. Similarly, Lemma \ref{lem:htostrvsvpt} also holds when
restricted to hedge-to-hedge transductions (the same constructions apply):

\begin{lemma}\label{lem:htohtrvswnvpt}
    $\wnvpt(\Sigma,\Delta) \equiv \htohtr(\Sigma_c\mathord\times \Sigma_r,\Delta)$.
\end{lemma}

\paragraph{Removing the tail-recursive assumption}
As we have seen in the proof of Lemma~\ref{lem:htostrvsvpt},
the behavior of a \vpt is naturally encoded by a tail-recursive \htos. Intuitively, the word $w_3$
of rules of \htos should be produced after having processed the whole hedge.

We prove now that if we run \vpts on the \fcns encoding of hedges,
then we can express any \htos-definable transduction. Intuitively, in the \fcns encoding,
the return symbol of the root of the first tree of the hedge is encountered at the end of the processing of the hedge. As a consequence, the word $w_3$ can be output when processing this symbol. Formally, we have:

\begin{lemma}\label{lem:fcns-htostrvsvpt}
    $\vpt(\Sigma_\bot,\Delta) \circ \fcns \equiv \htos(\Sigma_c\mathord\times \Sigma_r,\Delta)$. 
\end{lemma}

\begin{proof}[(Sketch)]
A construction similar to the one presented in the proof of Lemma~\ref{lem:htostrvsvpt},
based on pairs of states of the \vpt,
can be used to build an equivalent \htos. It will not necessarily be tail-recursive as the output of the return transition will be produced last. Note also that to handle empty subtrees encoded 
by $\bot_c\bot_r$, the resulting \htos may associate a non-empty output word to leafs.
It is however not difficult to simulate such rules.

Conversely, the construction is a bit more complex. States of the \vpt store the rule that is applied at the previous level, and the position in this rule (beginning, middle, or end). A special case is that of the first level, as there is no previous level. In this case, we store the initial state we started from. This information is stored in the stack, so as to recover it and faithfully simulate the
application of the rule. The case of rules associated with leafs is handled using the $\bot_c,\bot_r$ symbols, and dedicated rules. Details can be found in the Appendix.
\end{proof}

\begin{lemma}\label{lem:fcns-htohtrvswnvpt}
    $\wnvpt(\Sigma_\bot,\Delta)\circ \fcns \equiv \htoh(\Sigma_c\mathord\times \Sigma_r,\Delta)$.
\end{lemma}

\subsection{Comparison with other tree transducer models} 

\htos correspond to parameter-less macro forest transducers \cite{PerSei04} without swapping nor copying, that output strings. Therefore by Lemma
\ref{lem:htostrvsvpt}, \vpts are strictly less expressive than mfts. Macro tree transducers (mtts) are transducers on ranked trees \cite{Engelfriet03}. 
To compare them with \vpts, which run on (linearization of) hedges, we use the first-child next-sibling encoding. 
As shown in \cite{PerSei04}, any mft is equivalent to the composition of two mtts on those encodings. 
Linear-size increase transformations (or transducers) are those transformations such that 
the size of an output is linearly bounded by the size of the input. 
In \cite{DBLP:conf/fsttcs/Maneth03} it is shown that any linear-size increase transformation
defined by an arbitrary composition of mtts is definable by a single linear-size increase mtt. 
Therefore, linear-size increase mfts are equivalent to linear-size increase mtts. Since \vpts 
clearly define linear-size increase transformations, they are also strictly included in mtts.

Top-down ranked tree transducers with the linear and non-swapping restrictions are equivalent to \htob transducers on first-child next-sibling encodings. By Lemma \ref{lem:htobtrvshth},
we get that they are strictly less expressive than \wnvpts, and therefore \vpts. The arguments on the size of the ouputs in the proof of Lemma \ref{lem:htobtrvshth} still applies when dropping 
that restriction (the yield transduction cannot be defined), and therefore top-down ranked tree transducers are incomparable with \vpts. For the same reasons, bottom-up tree transducers 
are also incomparable with \vpts.

Finally, let us mention the \emph{uniform tree transducers} introduced by Neven and Martens \cite{conf/icdt/MartensN03}, and inspired by the XSLT language. These transducers can duplicate subtrees, but must use the same state
to transform all the children of a node. For those reasons they are incomparable with \vpts \cite{Servais11}.

\appendix

\label{sec:appendix}

\section{Appendix: Proof of Lemma 3 and 4}

\begin{proof}
 % Let $\Lambda$ and $\Delta$ be two finite alphabets. 
 Let $A=(Q,I,F,\Gamma,\delta)\in \vpt(\Sigma,\Delta)$. We
  define  $T=(Q',I',\delta')\in \htostr(\Sigma_c\mytimes\Sigma_r,\Delta)$ as follows:
\begin{itemize}
\item $Q'=\{(q_1,q_2)\in Q^2 \mid \exists w \in \wnset_{\Sigma} \textup{ s.t. }(q_1,\bot) \xrightarrow{w}(q_2,\bot)\}$
\item $I'=Q' \cap (I\times F)$
\item for all $c \in \Sigma_c$, $r \in \Sigma_r$, and all states
  $p_1,p_2,q_1,q_2,q'_1$ such that $(q_1,q_2),(p_1,p_2),(q'_1,q_2)\in
  Q'$, if there exist transitions $(q_1,c,\gamma,w_1, p_1) \in
  \delta_c$, $(p_2,r,\gamma,w_2,q'_1) \in \delta_r$, we build the rule:
$$
(q_1,q_2)((c,r)(x_1)\cdot x_2) \rightarrow w_1 \cdot (p_1,p_2)(x_1) \cdot w_2 \cdot (q'_1,q_2)(x_2)
$$
In addition, we also have:
$$(q_1,q_2)(0) \rightarrow \epsilon \in \delta' \iff q_1=q_2$$
\end{itemize}
It can be shown by induction that for all well-nested words $w\in \wnset_\Sigma$, $T$ has a
computation over $\hedge(w)$ leading to $(q_1,q_2)$ producing $w'$ iff 
$A$ admits a run from $(q_1,\bot)$ to $(q_2,\bot)$ over $w$ producing $w'$. 
Observe also that by definition $T$ is tail-recursive.

Notice that if $A$ is a \wnvpt, then we 
have $w_1w_2\in\wnset_\Sigma$, and thus $T\in \htoh$.
This proves one direction of Lemma~\ref{lem:htohtrvswnvpt}.

\medskip

Conversely, let us consider the transducer $T=(Q,I,\delta)$ from $\htostr(\Sigma_c\mytimes\Sigma_r,\Delta)$. We
define $A=(Q',I',F',\Gamma',\delta')\in \vpt(\Sigma,\Delta)$ as follows: $Q'=Q$, $I'=I$, $F' = \{q\in Q \mid q(0) \rightarrow \epsilon \in \delta \}$, $\Gamma' = \delta$ and
for every rule $ t = q((c,r)(x_1)\cdot x_2)\rightarrow w_1 q_1(x_1) w_2 q_2(x_2)\in \delta$, we add the following rules to $\delta'$:
$$
(q,c,t,w_1,q_1) \qquad \{(q',r,t,w_2,q_2) \mid q'\in F'\}
$$

It can be shown by induction that
for all well-nested word $w\in \wnset_\Sigma$, $B$ has a
computation over $\hedge(w)$ leading to $q$ producing $w'$ iff 
$A$ admits a run from $(q,\bot)$ to $(q',\bot)$ over $w$ producing $w'$, for some $q'\in F'$. 

Notice that that if $B\in \htoh$, then we 
have $w_1w_2\in\wnset_\Sigma$, and thus $A$ is a well-nested \vpt. This proves the other direction of Lemma~\ref{lem:htohtrvswnvpt}.
\end{proof}

% \section{Proof of Lemma~\ref{lem:fcns-htostrvsvpt} and~\ref{lem:fcns-htohtrvswnvpt}}

\section{Appendix: Proof of Lemma 5 and 6}

\begin{proof}
 Let $A=(Q,I,F,\Gamma,\delta)\in \vpt(\Sigma_\bot,\Delta)$. We first define the two following sets:

 $$
 \begin{array}{ll}
 X           {=}  \{(p,q)\in Q^2 | & \!\!\!\!\!\!\text{there exists a run }(p,\bot) \xrightarrow{cwr} (q,\bot) \text{ in }A, \text{ with }c\in \Sigma_c, r\in \Sigma_r, w\in\wnset_{\Sigma_\bot}\}\\
 X_{\bot} {=} \{(p,q)\in Q^2 | & \!\!\!\!\!\text{there exists a run
 }(p,\bot) \xrightarrow{\bot_c\bot_r} (q,\bot) \text{ in }A\}
 \end{array}
 $$
  We define  $B=(Q',I',\delta')\in \htos(\Sigma_c\mytimes\Sigma_r,\Delta)$ as follows:
$Q'=X\cup X_\bot$, $I'=X \cap (I\times F)$, for all $(p,q)\in X_\bot$, and all transitions $(p,\bot_c,\gamma, w, p')$, $(p', \bot_r, \gamma, w', q)$, we add the following rule to $\delta'$:
$(p,q)(0)\rightarrow ww'$.

In addition, for every $c \in \Sigma_c$, $r \in \Sigma_r$, and for every states
  $p,q,p_1,p_2,p_3$ such that $(p,q)\in X$, and $(p_1,p_2),(p_2,p_3)\in
  Q'$, if there exist a transition $(p,c,\gamma,w_1, p_1) \in
  \delta_c$ and a transition $(p_3,r,\gamma,w_3,q) \in \delta_r$, we build the rule:
$$
(p,q)((c,r)(x_1)\cdot x_2) \rightarrow w_1 \cdot (p_1,p_2)(x_1) \cdot (p_2,p_3)(x_2) \cdot w_3
$$

It can be shown by induction that for all well-nested word $w\in \wnset_\Sigma$, $B$ has a
computation over $\hedge(w)$ leading to $(q_1,q_2)$ producing $w'$ iff 
$A$ admits a run from $(q_1,\bot)$ to $(q_2,\bot)$ over $\fcns(w)$ producing $w'$. 

Observe also that $B$ does not comply with the definition of \htos as the first set of rules 
may produce non-empty words. However, it is easy to transform $B$ to ensure
this property as follows: for every rule $(p,q)(0)\rightarrow x$, build a state $(p,x,q)$,
and add the rule $(p,x,q)(0)\rightarrow \epsilon$. Then, modify the second set of rules
by replacing $(p,q)$ by $(p,x,q)$, and introducing $x$ at the convenient position in
the output of the rule. Note that this transformation will result in non-empty "$w_2$" words.

In addition, we assumed that input words are non-empty. As a consequence, the \fcns encodings
considered as input are different from the word $\bot_c\bot_r$. This justifies that the initial
states can be taken in $X$ only. This also implies that the removing of non-empty leaf rules
described before is correct, as every leaf rule will be applied in the context of some rule associated 
with an internal node.

Last, for the proof of Lemma~\ref{lem:fcns-htohtrvswnvpt}, it is easy to verify that if $A$ is a \wnvpt, 
then $B\in \htoh$.

\medskip

Conversely, let us consider the transducer $B=(Q,I,\delta)$ in
$\htos(\Sigma_c\mytimes\Sigma_r, \Delta)$. We
define $A=(Q',I',F',\Gamma',\delta')$ in $\vpt(\Sigma_\bot,\Delta)$ as
follows:

 $Q'=\{(q,i)\mid q\in I, i\in \{0,1\}\}\cup \{(t,i)\mid t \in\delta, i\in\{0,1,2\}\}\cup \{q_\bot\}$, $I'=I\times\{0\}$, $F' = I\times\{1\}$, 
$\Gamma' = Q'$, and  for every rule $ t = q((c,r)(x_1)\cdot x_2)\rightarrow w_1 q_1(x_1) w_2 q_2(x_2) w_3\in \delta$ such that
$q\in I$, we add the two following rules to $\delta'$:
$$
((q,0),c,(q,0),w_1,(t,0)) \qquad ((t,2),r,(q,0),w_3,(q,1))
$$

\medskip

In addition, for every two rules 
$$
\begin{array}{llll}
 t & = & q((c,r)(x_1)\cdot x_2)\rightarrow w_1 q_1(x_1) w_2 q_2(x_2) w_3\in \delta \\
t' & = & q'((c',r')(x_1)\cdot x_2)\rightarrow w'_1 q'_1(x_1) w'_2 q'_2(x_2) w'_3\in \delta 
\end{array}
$$
and $i\in\{0,1\}$ such that $q'_i=q$,  we add the two following rules to $\delta'$:
$$
((t',i),c,(t',i),w_1,(t,0))$$
$$ ((t,2),r,(t',i),w_3x,(t',i+1)) \text{ where }
x=\left\{
\begin{array}{ll}
w'_2 & \text{if }i=0\\
\epsilon & \text{otherwise}
\end{array}
\right.
$$

\medskip

Last, we consider rules associated with leafs: for every rule $q(0)\rightarrow \varepsilon$,
we add the two following transitions: (provided that the $i$-th state of the rule $t$ is $q$)
$$
((t,i),\bot_c,(t,i),\varepsilon,q_\bot)
\qquad
(q_\bot,\bot_r,(t,i),\varepsilon,(t,i+1))
$$

It can be shown by induction that 
for all well-nested word $w\in \wnset_\Sigma$, $B$ has a
computation over $\hedge(w)$ leading to $q$ producing $w'$ iff the two following 
properties are verified:
\begin{itemize}
\item if $q\in I$, then $A$ admits a run from $(q,0)$ to $(q,1)$ over $\fcns(w)$ producing $w'$
\item for every $(t,i)\in \delta \times\{0,1,2\}$ such that $t= p((c,r)(x_1)\cdot x_2)\rightarrow w_1 q_1(x_1) w_2 q_2(x_2) w_3\in \delta$ and $q_i=q$, 
$A$ admits a run from $(t,i)$ to $(t,i+1)$ over $\fcns(w)$ producing $w'x$, 
where $x=\epsilon$ if $i=1$, and $x=w_2$ otherwise.
\end{itemize}
\end{proof}

\vspace{2mm}
\textbf{Acknowledgments} We are very grateful to Sebastian Maneth and anonymous referees for helpful comments on this paper. 

\bibliographystyle{eptcs}
\bibliography{biblio}
\end{document}